\documentclass{article}

\RequirePackage{fix-cm}

\usepackage{setspace}

\usepackage[misc,geometry]{ifsym} 

\usepackage{color}
\usepackage{srcltx}
\usepackage{fancyhdr}
\usepackage{latexsym}
\usepackage{amsfonts}
\usepackage{amscd}
\usepackage{epsfig}
\usepackage[cp850]{inputenc}
\usepackage[english]{babel}
\usepackage{ifthen}
\usepackage{float}
\usepackage{amsmath}
\usepackage{amssymb}
\usepackage{mathrsfs}
\usepackage[mathscr]{eucal}
\usepackage{dsfont}
\usepackage{amstext}
\usepackage{syntonly}

\usepackage{tikz} 
\usepackage{color}
\usepackage{graphicx}

\usepackage{enumerate}
\usepackage{amssymb}
\usepackage[english]{babel}
\usepackage{amsthm}
\usepackage{amstext}
\usepackage{graphicx}
 \usepackage{mathptmx}      
\usepackage{latexsym}

  \usepackage{enumerate}

\newtheorem{theorem}{Theorem}[section] 
\newtheorem{lemma}[theorem]{Lemma}     
\newtheorem{corollary}[theorem]{Corollary}
\newtheorem{proposition}[theorem]{Proposition}
\newtheorem{definition}[theorem]{Definition}   
\newtheorem{remark}[theorem]{Remark}

\usepackage{color}

\usepackage{enumerate}
\usepackage{amssymb}
\usepackage[english]{babel}
\usepackage{amstext}
\usepackage{graphicx}
 \usepackage{mathptmx}      
\usepackage{latexsym}
\newcommand{\footremember}[2]{%
    \footnote{#2}
    \newcounter{#1}
    \setcounter{#1}{\value{footnote}}%
}

\date{}
\author{ Asier Estevan\footremember{trailer}{Dpto. Estad\'istica, Inform\'atica y Matem\'aticas, Instituto INAMAT, Universidad P\'ublica de Navarra.   Campus Arrosad\'{\i}a, 31006.  Iru\~na-Pamplona, Navarra, Spain.    }%
  }
\usepackage{color}
\usepackage{graphicx}

\usepackage{color}

\title{A Debreu's Open Gap Lemma for semiorders}

\begin{document}

%
%


 \maketitle


 \begin{abstract} 
 
 
 The problem of finding a (continuous) utility function for a semiorder has been  studied
since  in 1956 R.D. Luce introduced in \emph{Econometrica} the notion. 
There was almost no results on the  continuity of the representation. A similar result to Debreu's Lemma, but for semiorders, was never achieved. Recently, some necessary conditions for the existence of a continuous representation as well as some conjectures were presented by A. Estevan. 
 In the present paper we prove these conjectures, achieving the desired version of  Debreu's Open Gap Lemma for bounded semiorders. This result allows to remove the open-closed and closed-open gaps of a subset $S\subseteq \mathbb{R}$, but now keeping the constant threshold, so that $x+1<y$ if and only if $g(x)+1<g(y) \, (x,y\in S)$. Therefore, the  continuous representation (in the sense of Scott-Suppes) of bounded semiorders is characterized. These results are achieved thanks to the key notion of $\epsilon$-continuity, which generalizes the idea of continuity for semiorders.
 

  \end{abstract}


 \section{Introduction}\label{s1} 

 In 1964 the Nobel laureate Gerard Debreu  proved in his famous  \emph{Debreu's Open Gap Lemma}   that, given any subset  $S\subseteq \mathbb{R}$,
then there is a strictly increasing function
$g\colon S\to \mathbb{R}$ such that all the gaps of
$g(S)$ are open. 
 Unfortunately, 
 to this day it is unknown when this strictly increasing  function $g$ exists  when we also impose to $g$ to satisfy the geometrical condition
 $x \prec y \Leftrightarrow g(x) + 1 < g(y)$, for every $x,y \in X$. \cite{D2}

The solution to the question before would solve the problem of the characterization of the existence of a continuous SS-representation of semiorders.
Although the notion of \emph{semiorder} is usually attributed to Luce in 1954, it was first  introduced by Wiener (as well as the concept of \emph{interval order}). 
 \cite{c,a,b} The notion of interval order generalize the idea of semiorder, and it was studied deeply by Fishburn in the 1970's. \cite{8,9,f1,f2} 
The use of both relations was due to the need of working with situations of intransitive indifference.

In the present paper we focus on semiorders and prove some conjectures introduced in \cite{liburu} related to the existence of a continuous representation (\emph{in the sense of Scott-Suppes} or \emph{SS-representation}, for short) for semiorders. As a result, we achieve a version  of  Debreu's Open Gap Lemma for bounded semiorders. The representability problem for semiorders was finally solved in \cite{CISS} (see also \cite{GgureaJ}), and the present paper closed the question on continuity of SS-representations for bounded semiorders.

\medskip

An asymmetric binary relation $\prec $ on $X$   is said to be a \emph{semiorder} if the following two conditions are satisfied:\\
\noindent$(1)$ $(x \prec y) \wedge (z \prec t)  \Rightarrow (x \prec t) \vee (z \prec y), \quad x,y,z,t \in X,$\\ 
\noindent$(2)$ $(x \prec y) \wedge {(y \prec z)}  \Rightarrow (x \prec w) \vee (w \prec z), \quad x,y,z,w \in X.$

A semiorder is said to be \emph{bounded} if there are no strictly increasing or decreasing infinite sequences, i.e., there is no sequence  $(x_n)_{n \in \mathbb{N}}\subseteq X$ such that $ \cdots \prec x_{n+1} \prec x_n \prec \cdots \prec x_1$ or $x_1\prec \cdots \prec x_n \prec x_{n+1} \prec \cdots $.

Due to condition $(1)$, a semiorder is a particular case of an \emph{interval order}.
Its symmetric part is denoted by $\precsim$, so that $a \precsim b  \Leftrightarrow \neg(b \prec a)$. 
The \emph{indifference} \rm  relation $\sim$ associated to $\prec$ is defined by 
 $a \sim b  \Leftrightarrow (a \precsim b) \wedge (b \precsim a).$
 It is well known that $\precsim$ and $\sim$ may fail to be transitive. \cite{8,9,Luc,SS}   
   A \emph{preorder} $\precsim$ on $X$ is a binary relation  which is reflexive and transitive. An antisymmetric preorder is said to be an \emph{order}. A \emph{total preorder} \rm $\precsim$ on a set $X$ is a preorder such that if $x,y \in X$ then $(x \precsim y) \vee (y \precsim x)$ holds. In the case of preorders, it is well known that the corresponding indifference is transitive, i.e., it is an equivalence relation.

The SS-representation    is defined by means of a single function  $u{ { }\colon}X \rightarrow \mathbb{R}$  such that $x \prec y \Leftrightarrow u(x) + 1 < u(y)$, for every $x,y \in X$. \cite{CISS,Luc,SS,vinc} 
When the set $X$ is   endowed with a topology $\tau$,   the  semicontinuity  or  continuity   of the numerical representations (if any) is also studied. \cite{Campi, CIZ,Ge}

The analogous problem related to the existence of a continuous representation but now for total preorders was solved by the   Gerard Debreu in 1964. \cite{D2}  For this purpose, 
 a \emph{lacuna} of $S\subseteq \mathbb{R}$ was defined as a non-degenerate interval disjoint from $S$ and having a lower bound and an upper bound in $S$, and a \emph{gap} of $S$ as a maximal lacuna of $S$.
The famous \emph{Debreu's Open Gap Lemma}  reads as follows:  \cite{D2}

\begin{lemma}(\textbf{Open Gap Lemma})
If $S\subseteq \mathbb{R}$, then there is a strictly increasing function
$g\colon S\to \mathbb{R}$ such that all the gaps of
$g(S)$ are open. 
\end{lemma}

 Unfortunately, there is no answerd when we also impose to $g$ to satisfy the geometrical condition $x \prec y \Leftrightarrow g(x) + 1 < g(y)$, for every $x,y \in X$, although it is known that this function $g$ may fail to exist. \cite{ij2013}

 \bigskip
 
  In order to approach to the desired results and construct a theorem such as Debreu's Open Gap Lemma but for semiorders, 
 the concept of $\epsilon$-continuity was succesfully introduced as a generalization on the idea of continuity for semiorders. \cite{liburu} 
In the case of semiorders,  there is an invariant threshold $k$ in the SS-representations (we may assume that $k=1$) that allows us to  compare the  length of each jump-discontinuity with this value $k=1$. Hence, it makes perfect sense to say that a semiorder is $r$-continuous (with $r>0$) if there exists a SS-representation   $(u,1)$ such that the length of each jump-discontinuity is bounded by this constant $r$. Then, we can approach to the idea of the usual continuity just tending $r$ to 0.
 Through this idea some conjectures were proposed. In the present paper we present the corresponding proofs.
 
 For more details  on this subject we suggest  the readings \cite{Ales,BR}. In particular,  for those notions related to the continuous representability of semiorders and $\epsilon$-continuity, we efusevely recommend to read \cite{liburu}. 

 \medskip
 
 \emph{The structure of the paper goes as follows:}
First, in Section~\ref{s2}, necessary conditions for the existence of a continuous SS-representation are collected and the  concept of $\epsilon$-continuity is presented. Then,   the image subset $u(X)$ is studied for a given SS-representation $(u,1)$ of a semiorder that satisfies the aforementioned necessary conditions. In next Section~\ref{sD}, some conjectures on the continuous SS-representability of semiorders are proved.
 Finally, in Section~\ref{sDeb}, the results obtained before are summarized but now abstracted from the context of semiorders, just as a version of the \emph{Debreu's Open Gap Lemma} but with additional component of a \emph{threshold}.

 \setcounter{footnote}{0}
 \renewcommand{\thefootnote}{\arabic{footnote}}

 \section{Necessary conditions for SS-representability}\label{s2}  

In the present paper we focus on the existence of continuous SS-representations. Hence, we shall assume that the semiorder is SS-representable. The characterization of the existence a SS-representation is known and it is made by means of \emph{s-separability} and \emph{regularity}.
 \cite{GgureaJ,CISS}

 \begin{theorem} \label{main} Let $X$ be a nonempty set. Let $\prec$ be a typical semiorder defined on $X$. Then, $\prec$ is representable in the sense of Scott and Suppes if and only if it is both s-separable and regular with respect to sequences. 
 \end{theorem} 
 
Now we collect a few concepts, results and those necessary conditions for the existence of a continuous SS-representation which are available in literature. \cite{base,BR,liburu,ij2013}.


 \begin{theorem}Let $\prec$ be an interval order defined on a set $X$. Then the indifference $\sim^0$ associated to the main trace is an equivalence relation. \end{theorem}

 \begin{definition}\rm
  Let $\prec$ be an interval order defined on a topological space $(X,\tau)$.  The  topology $\tau$ is said to be \emph{compatible with respect to the indifference of the main trace of $\prec$} if $x \sim^0 y \Rightarrow (x \in \mathcal{O} \iff y \in \mathcal{O})$ holds true for every $x,y \in X$ and every $\tau$-open subset $\mathcal{O} \in \tau$.%
  \end{definition}
 
 
In particular, in the main case in  which $x\sim_0 y\iff x=y$, 
 the topology is always compatible. 
 
Next proposition may be found in \cite{base}. 
 
 \begin{proposition} \label{perp} Let $(X,\tau)$ be a topological space endowed with a semiorder $\prec$. Assume that $\tau$ is compatible with respect to the indifference of the main trace of $\prec$. Suppose also that $\prec$ is representable in the sense of Scott and Suppes by means of a pair $(u,1)$ with $u$ continuous. Then the total preorder $\precsim^0$ is $\tau$-continuous. \end{proposition}

From now on, we shall assume that  the topology of the space is compatible with respect to the indifference of the main trace of the semiorder (e.g. the quotient set $X\slash \sim^0$ coincides with $X$). Hence, by Proposition~\ref{perp}, we will assume that $\precsim^0$ is $\tau$-continuous.

With respect to  continuity,    the following result were proved too. \cite{liburu, ij2013}
 
 \begin{lemma} \label{lok} Let $(X,\tau)$ be a topological space endowed with a semiorder $\prec$. Assume that $\prec$ is representable in the sense of Scott and Suppes by means of a pair $(u,1)$ with $u$ continuous.  Then the following properties hold true:
 \begin{itemize}
 \item[\emph{(a)}] The semiorder $\prec$ is $\tau$-continuous. 
  \item[\emph{(b)}] If a net $(x_j)_{j \in J} \subseteq X$\footnote{$J$ denotes here a directed set of indices. Since this does not lead to confusion, we will use the same notation `$<$' of the order on the real numbers than for the partial order on the set of indices $J$.} converges to two points $a,b \in X$, then $a \sim^0 b$. 
 \item[\emph{(c)}] If a net $(x_j)_{j \in J} \subseteq X$ converges to $a \in X$, and $b,c \in X$ are such that $x_j \prec b \precsim a$ and also $x_j \prec c \precsim a$ for every $j \in J$, then $b \sim^0 c$. 
 \item[\emph{(d)}] If a net  $(x_j)_{j \in J} \subseteq X$ converges to $a \in X$, and $b,c \in X$ are such that $a \precsim b \prec x_j $ and also $a \precsim c \prec x_j $ for every $j \in J$, then $b \sim^0 c$.

\item[\emph{(e)}]Let $(x_j)_{j\in J}$ and $ (y_k)_{k\in K}$ be two nets such that they converge to the same point $a$ in $X$ and $(w_r)_{r\in R}$ (respectively $(z_s)_{s\in S}$) two $n$-adjoint nets of $(x_j)_{j\in J}$ (respectively, of $(y_k)_{k\in K}$) for some $n\in \mathbb{Z}-\{0\}$.
If there are two elements $b,c\in X$ such that
 $w_{r}\prec^0 b\prec^0 z_{s}$ as well as $w_{r}\prec^0 c\prec^0 z_{s}$ (for each $r\in R, s\in S$), then $b\sim^0 c.$ 
 
 Furthermore:
 
 \begin{enumerate}
 \item[$(e_1)$]
 If there exist two pairs of  adjoint nets $(z_s)\preccurlyeq^n (y_i)$  and $(w_r)\preccurlyeq^{n-1} (x_j)$   (dually, $(z_s)\preccurlyeq^{-n} (y_i)$ and $(w_r)\preccurlyeq^{-n+1} (x_j)$) such that   $(x_j)_{j\in J}$ and $ (y_i)_{i\in I}$  converge to the same point $a$ in $X$, and if there is a point $c$ such that $z_s\prec_0 c \prec w_r$ (resp.,  $w_r\prec c\precsim z_s$), then there is no $c_1$ such that  $z_s\precsim c_1 \prec c$ (resp., $ c \prec c_1\precsim z_s$).
  
  \item[$(e_2)$]
 If there exist two pairs of  adjoint nets $(z_s)\preccurlyeq^{n-1} (y_i)$  and $(w_r)\preccurlyeq^{n} (x_j)$   (dually, $(z_s)\preccurlyeq^{-n+1} (y_i)$ and $(w_r)\preccurlyeq^{-n} (x_j)$) such that   $(x_j)_{j\in J}$ and $ (y_i)_{i\in I}$  converge to the same point $a$ in $X$ with $y_i\precsim_0 a \precsim_0 x_j$, and if there is a point $c$ such that $z_s\precsim c \prec_0 w_r$ (resp.,  $w_r\prec_0 c\precsim z_s$), then there is at most one $c_1$ such that  $c\precsim c_1 \prec w_r$ (resp., $ w_r \prec c_1\precsim c$).
 
 \end{enumerate}

 \end{itemize}
  \end{lemma}
  
  Throughout the paper, we shall refer to these conditions $(a)-(e)$ by (NC) (\emph{necessary conditions}). The last condition $(e)$\footnote{Subconditions $(e_1)$ and $(e_2)$ have been presented here for first time, however, since the technique of the proof is similar to that used in $(e)$, it is left to the reader.} is presented by means of \emph{$n$-adjoint nets}. This concept was introduced for first time in \cite{liburu}, as well as the following result:
  
\begin{lemma}\label{Ladj}
Let $\prec$ be a semiorder defined on a topological space $(X,\tau)$ and let $(u,1)$ be a continuous representation.
 If $(x_j)_{j\in J}$ and $ (y_k)_{k\in K}$ are  $n$-adjoint nets, then $\lim_{j\in J} u(x_j)+n=\lim_{k\in K} u(y_k)$.
\end{lemma}

It was already highlighted that 
conditions (c) and (d) are now particular cases of  condition $(e)$, in which $n=1.$ See  \cite{liburu} for more detail.

We recover  
the following corollary that summarizes the structure   of any  SS-repre\-sen\-ta\-tion (i.e. not necessarily continuous) of a continuously representable semiorder, i.e., of a semiorder that satisfied  the necessary conditions (NC).

\begin{corollary}\label{Cexp}
Let $\prec$ be a continuously representable semiorder on $(X,\tau)$, i.e.  satisfying the necessary conditions (NC). Let $(u,1)$ be a SS-representation. Suppose that there is a discontinuity at a point $a$ such that $[r,u(a))$ (or  $(u(a), r]$) is a bad gap (we call this initial assumption as \emph{initial condition} (ic)).

Then, $u(X)\cap [r+1,u(a)+1]$ (or $u(X)\cap [u(a)-1, r-1]$, respectively) has at most one point $s_1\in \mathbb{R}$ and one of the following situations holds:

\begin{itemize}
\item[$(a_1)$] Depending on the existence of that point $s_1$, it holds that $[r+1,u(a)+1]$ (there is no $s_1$), $(r+1,u(a)+1]$ ($s_1=r+1$) or $[r+1,u(a)+1)$ ($s_1=u(a)+1$) is a gap, or $[r+1,u(a)+1]$ is the union of two consecutive gaps $[r+1,s_1)\cup (s_1,u(a)+1]$. 
In that case, $u(X)\cap [r+2,u(a)+2]$ (resp. $u(X)\cap [u(a)-2, r-2]$) has at most one point $s_2\in \mathbb{R}$ (such that $s_2\leq s_1+1$) and we will continue applying these cases $(a_1)$ or $(b_1)$, but now on $[r+2,u(a)+2]$ (resp. $ [u(a)-2, r-2]$).

\item[$(b_1)$]  
If the case $(a_1)$ does not hold, then there exist $\gamma_l, \gamma_r\geq 0$, at least one of them is strictly positive, such that $u(X)\cap [r+1-\gamma_l, u(a)+1+\gamma_r]$ 
  (or $u(X)\cap [u(a)-1-\gamma_l, r-1+\gamma_r]$, respectively) has at most that point $s_1$.

  
  \begin{enumerate}
  \item[$(b_{11})$]   In case this point $s_1$ exists, then:
    \begin{enumerate}
  \item[$(b_{111})$]   If  $\delta_r=0$, then $(s_1+1,u(a)+2]=\emptyset$,  and  we shall continue arguing to the right with $(s_1+1,u(a)+2]$  as a gap of case (ci). 
  \item[$(b_{112})$]   If   $\delta_l=0$, then $[r+2,s_1+1]$ may contain a unique point $s_2$.  We shall continue arguing to the right with $[r+2,s_2)$ as a gap of case (ci), in case $s_2$ exists, and without restrictions otherwise. 
  \item[$(b_{113})$]   If   $\delta_l>0$ and $\delta_r>0$, then we continue arguing to the right without restrictions.
  \end{enumerate}
  
  \item[$(b_{12})$]   In case there is no point $s_1$, then $ [r+2,u(a)+2]$ (resp. $ [u(a)-2, r-2]$)  may contain more than one point and we continue arguing to the right (resp. left) without restrictions. 
  \end{enumerate}

\end{itemize}

Moreover, $u(X)\cap [r-1,u(a)-1)=\emptyset$ (or $u(X)\cap (u(a)+1, r+1]=\emptyset$, resp.). Here, again,  one of the following situations holds:

\begin{itemize}
\item[$(a_2)$]  If $[r-1,u(a)-1]$ or $[r-1,u(a)-1)$ (resp. with $[u(a)+1, r+1]$ or $[u(a)+1, r+1)$) is a gap (depending on the existence of that  adjoint point $u(a)-1$ (resp. $u(a)+1$)), then $ u(X)\cap [r-2,u(a)-2]$ (resp. $u(X)\cap [u(a)+2, r+2]$) has at most one point $s$, which is in fact the adjoint point $u(a)-2$ (resp. $u(a)+2$) in case $u(a)-1$ (resp. $u(a)+1$) exists, and we will continue applying these cases $(a_2)$ or $(b_2)$ but now on $ [r-2,u(a)-2]$ (resp. $ [u(a)+2, r+2]$).

\item[$(b_2)$]  If the case $(a_2)$ before does not hold, 
then there exist  $\gamma_l, \gamma_r\geq 0$, at least one of them is strictly positive, 
 such that $u(X)\cap [r-1-\gamma_l, u(a)-1+\gamma_r]$  (or $u(X)\cap [u(a)+1-\gamma_l, r+1+\gamma_r]$, respectively)
  has at most that point $u(a)-1$.
   \begin{enumerate}
  \item[$(b_{21})$]   In case this point $s_1=u(a)-1$ exists, then:
    \begin{enumerate}
  \item[$(b_{211})$]   If  $\delta_l=0$, then $u(X)\cap [r-2,u(a)-2)=\emptyset$,   but notice that this case belongs to $(a_2)$.
  \item[$(b_{212})$]   If $\delta_l>0 $, then $[r-2, u(a)-2]$ may be nonempty, and we continue to the left without restrictions. 
  
  \end{enumerate}
  
  \item[$(b_{22})$]   In case there is no point $s_1=u(a)-1$, then $u(X)\cap [r-2,u(a)-2]$   may contain more than one point and we continue arguing to the left without restrictions.
  \end{enumerate}
 
\end{itemize}

\end{corollary}
\begin{proof}
First,  if there is a discontinuity at a point $a$ such that $[r,u(a))$ (or  $(u(a), r]$) is a gap, then  there is a net $(u(y_i))_{i\in I}$ converging to $r$ in $\mathbb{R}$ and there is another net $(u(x_j))_{j\in J}$ (it may be constant, i.e. $u(x_j)=u(a)$ for any $j\in J$) converging to $u(a)$ in $\mathbb{R}$. 

 The first statement before points $(a_1)$ and $(b_1)$, relative to the possible  existence of a unique point in $u(X)\cap [r+1,u(a)+1]$ (or $u(X)\cap [u(a)-1, r-1]$, respectively), is proved in Proposition~4.9 in \cite{liburu}. 
 
   If $[r+1,u(a)+1]$ (resp. $[u(a)-1, r-1]$) is 
 as described in case $(a_1)$, then notice that there exist adjoint nets
 $(z_s)_{s\in S}$ and $(w_r)_{r\in R}$ such that $(y_i) \preccurlyeq (z_s)$ and  $(x_j) \preccurlyeq (w_r)$. Hence, by Proposition~4.10 in \cite{liburu}, 
  we deduce that $ u(X)\cap [r+2,u(a)+2]$ (resp. $ [u(X)\cap u(a)-2, r-2]$) 
 contains at most one point $s'$.  We will continue arguing on  $ [r+2,u(a)+2]$ (resp. $ [u(a)-2, r-2]$).

If $[r+1,u(a)+1]$ (resp. $[u(a)-1, r-1]$) is as described in  case $(b_1)$, i.e.,  there exist $\gamma_l, \gamma_r\geq 0$, at least one of them is strictly positive, such that $u(X)\cap [r+1-\gamma_l, u(a)+1+\gamma_r]$ 
  (or $u(X)\cap [u(a)-1-\gamma_l, r-1+\gamma_r]$, respectively) has at most that point $s_1$, then    
 the first three  possible cases (in which $s_1$ does exist) are directly deduced from condition $(e)$.  In case there is no point $s_1$, then (NC) says nothing about restrictions.

  The second part of the corollary is similarly proved.
  







\begin{figure}[htbp]
\begin{center}
\begin{tikzpicture}[scale=0.8]
\draw[thick] (-1.5,0) node[anchor=east] {$u(X)$};
    \draw[] (-1,0) -- (0.5,0);
    \draw[] (2,0) -- (3.5,0);
    \draw[] (5,0) -- (6.5,0);
    \draw[] (8,0) -- (9.5,0);

\draw[thick] (-0.75,0.4) node[anchor=east] {\small $0$};
\draw[thick] (0.75,0.4) node[anchor=east] {\small $0'5$};
\draw[thick] (3.75,0.4) node[anchor=east] {\small $1'5$};
\draw[thick] (6.75,0.4) node[anchor=east] {\small $2'5$};
\draw[thick] (0.75,0) node[anchor=east] {\small $)$};

\draw[thick] (5.25,0) node[anchor=east] {\small $($};
\draw[thick] (6.75,0) node[anchor=east] {\small $)$};
\draw[thick] (8.25,0) node[anchor=east] {\small $($};
\draw[thick] (3.75,0) node[anchor=east] {\small $)$};
\draw[thick] (-0.75,0) node[anchor=east] { $\bullet$};
\draw[thick] (2.25,0.4) node[anchor=east] {\small $1$};
\draw[thick] (2.25,0) node[anchor=east] { $\bullet$};
\draw[thick] (2.2,0) node[anchor=east] {\small $[$};

\draw[thick] (8.25,0.4) node[anchor=east] {\small $3$};
\draw[thick] (-0.85,0) node[anchor=east] {\small $[$};

\draw[dashed] (0.55,0) -- (2,-2);
\draw[dashed] (2,0) -- (2,-2);
\draw[dashed] (6.55,0) -- (8,-2);
\draw[dashed] (3.55,0) -- (5,-2);
\draw[dashed] (5,0) -- (5,-2);
\draw[dashed] (8,0) -- (8,-2);
\draw[thick] (5.25,0.4) node[anchor=east] {\small $2$};
\draw[thick] (4.75,0) node[anchor=east] {\small $\bullet$};
\draw[thick] (4.75,0.4) node[anchor=east] { $s$};
\draw[thick] (7.25,0) node[anchor=east] {\small $\bullet$};
\draw[thick] (7.25,0.4) node[anchor=east] { $s'$};

\draw[thick] (5.25,-2.4) node[anchor=east] {\small $2$};
\draw[thick] (5.25,-2) node[anchor=east] { $\bullet$};

\draw[thick] (-1.5,-2) node[anchor=east] {$g(u(X))$};

    \draw[] (-1,-2) -- (9.5,-2);

\draw[thick] (-0.75,-2.4) node[anchor=east] {\small $0$};
\draw[thick] (1.9,-2.4) node[anchor=east] {\small $0'5$};
\draw[thick] (2.15,-2) node[anchor=east] {\small $)$};
\draw[thick] (-0.75,-2) node[anchor=east] { $\bullet$};
\draw[thick] (2.25,-2.4) node[anchor=east] {\small $1$};
\draw[thick] (2.25,-2) node[anchor=east] { $\bullet$};
\draw[thick] (2.2,-2) node[anchor=east] {\small $[$};

\draw[thick] (8.25,-2.4) node[anchor=east] {\small $3$};
\draw[thick] (8.25,-2) node[anchor=east] { $\bullet$};

\end{tikzpicture}
\caption{Ilustration of case $(a_1)$ of Corollary~\ref{Cexp}, for the particular values of $r=0.5$ and $u(a)=1$.}
\label{figureCorollaryExamplea1}
\end{center}
\end{figure}

\begin{figure}[htbp]
\begin{center}
\begin{tikzpicture}[scale=0.8]
\draw[thick] (-1.5,0) node[anchor=east] {$u(X)$};
    \draw[] (-1,0) -- (0.5,0);
    \draw[] (2,0) -- (3.2,0);
    \draw[] (5.3,0) -- (9.5,0);

\draw[thick] (-0.75,0.4) node[anchor=east] {\small $0$};
\draw[thick] (0.75,0.4) node[anchor=east] {\small $0'5$};
\draw[thick] (3.85,0.4) node[anchor=east] {\tiny $1'5$};
\draw[thick] (5.25,0.4) node[anchor=east] {\tiny $2$};
\draw[thick] (6.75,0.4) node[anchor=east] {\small $2'5$};
\draw[thick] (0.75,0) node[anchor=east] {\small $)$};

\draw[thick] (5.55,0) node[anchor=east] {\small $($};
\draw[thick] (3.45,0) node[anchor=east] {\small $)$};
\draw[thick] (5.25,0) node[anchor=east] {\small $|$};
\draw[thick] (6.75,0) node[anchor=east] {\small $|$};
\draw[thick] (8.25,0) node[anchor=east] {\small $|$};
\draw[thick] (3.75,0) node[anchor=east] {\small $|$};
\draw[thick] (-0.75,0) node[anchor=east] { $\bullet$};
\draw[thick] (2.25,0.4) node[anchor=east] {\small $1$};
\draw[thick] (2.25,0) node[anchor=east] { $\bullet$};
\draw[thick] (2.2,0) node[anchor=east] {\small $[$};

\draw[thick] (8.25,0.4) node[anchor=east] {\small $3$};
\draw[thick] (-0.85,0) node[anchor=east] {\small $[$};

\draw[dashed] (0.55,0) -- (2,-2);
\draw[dashed] (2,0) -- (2,-2);
\draw[dashed] (3.55,0) -- (5,-2);
\draw[dashed] (5.25,0) -- (5.75,-2);
\draw[dashed] (3.25,0) -- (4.5,-2);
\draw[dashed] (5,0) -- (5,-2);
\draw[dashed] (8.25,0) -- (8.75,-2);
\draw[dashed] (6.25,0) -- (7.5,-2);
\draw[thick] (4.75,0) node[anchor=east] {\small $\bullet$};
\draw[thick] (4.75,0.4) node[anchor=east] { $s$};

\draw[thick] (5.25,-2.4) node[anchor=east] {\small $2$};
\draw[thick] (5.25,-2) node[anchor=east] { $\bullet$};

\draw[thick] (-1.5,-2) node[anchor=east] {$g(u(X))$};

    \draw[] (-1,-2) -- (9.5,-2);

\draw[thick] (-0.75,-2.4) node[anchor=east] {\small $0$};
\draw[thick] (1.9,-2.4) node[anchor=east] {\small $0'5$};
\draw[thick] (2.15,-2) node[anchor=east] {\small $)$};
\draw[thick] (-0.75,-2) node[anchor=east] { $\bullet$};
\draw[thick] (2.25,-2.4) node[anchor=east] {\small $1$};
\draw[thick] (2.25,-2) node[anchor=east] { $\bullet$};
\draw[thick] (2.2,-2) node[anchor=east] {\small $[$};

\draw[thick] (8.25,-2.4) node[anchor=east] {\small $3$};
\draw[thick] (8.25,-2) node[anchor=east] { $\bullet$};

\end{tikzpicture}
\caption{Ilustration of case $(b_{113})$ of Corollary~\ref{Cexp}, for the particular values of $r=0.5$ and $u(a)=1$.}
\label{figureCorollaryExampleb1}
\end{center}
\end{figure}

\begin{figure}[htbp]
\begin{center}
\begin{tikzpicture}[scale=0.8]
\draw[thick] (-1.5,0) node[anchor=east] {$u(X)$};
    \draw[] (-1,0) -- (0.5,0);
    \draw[] (2,0) -- (3.5,0);
    \draw[] (5,0) -- (6.5,0);
    \draw[] (8,0) -- (9.5,0);

\draw[thick] (-0.75,0.4) node[anchor=east] {\small $-2$};
\draw[thick] (0.75,0.4) node[anchor=east] {\small $-1.5$};
\draw[thick] (3.75,0.4) node[anchor=east] {\small $-0.5$};
\draw[thick] (6.75,0.4) node[anchor=east] {\small $0.5$};
\draw[thick] (0.75,0) node[anchor=east] {\small $)$};

\draw[thick] (5.25,0) node[anchor=east] {\small $($};
\draw[thick] (6.75,0) node[anchor=east] {\small $)$};
\draw[thick] (8.25,0) node[anchor=east] {\small $[$};
\draw[thick] (3.75,0) node[anchor=east] {\small $)$};
\draw[thick] (-0.75,0) node[anchor=east] { $\bullet$};
\draw[thick] (2.25,0.4) node[anchor=east] {\tiny $-1=s'$};
\draw[thick] (2.25,0) node[anchor=east] { $\bullet$};
\draw[thick] (2.2,0) node[anchor=east] {\small $[$};

\draw[thick] (8.25,0.4) node[anchor=east] {\small $1$};
\draw[thick] (-0.85,0) node[anchor=east] {\small $[$};

\draw[dashed] (0.55,0) -- (2,-2);
\draw[dashed] (2,0) -- (2,-2);
\draw[dashed] (6.55,0) -- (8,-2);
\draw[dashed] (3.55,0) -- (5,-2);
\draw[dashed] (5,0) -- (5,-2);
\draw[dashed] (8,0) -- (8,-2);
\draw[thick] (5.25,0.4) node[anchor=east] {\tiny $s=0$};
\draw[thick] (5.25,0) node[anchor=east] {\small $\bullet$};

\draw[thick] (5.25,-2.4) node[anchor=east] {\small $0$};
\draw[thick] (5.25,-2) node[anchor=east] { $\bullet$};

\draw[thick] (-1.5,-2) node[anchor=east] {$g(u(X))$};

    \draw[] (-1,-2) -- (9.5,-2);

\draw[thick] (-0.75,-2.4) node[anchor=east] {\small $-2$};
\draw[thick] (2.2,-2.4) node[anchor=east] {\small $-1$};
\draw[thick] (8.15,-2) node[anchor=east] {\small $)$};
\draw[thick] (-0.75,-2) node[anchor=east] { $\bullet$};
\draw[thick] (2.25,-2) node[anchor=east] { $\bullet$};
\draw[thick] (8.2,-2) node[anchor=east] {\small $[$};

\draw[thick] (8.25,-2.4) node[anchor=east] {\small $1$};
\draw[thick] (8.25,-2) node[anchor=east] { $\bullet$};

\end{tikzpicture}
\caption{Ilustration of case $(a_2)$ of Corollary~\ref{Cexp}, for the particular values of $r=0.5$ and $u(a)=1$.}
\label{figureCorollaryExamplea2}
\end{center}
\end{figure}

\begin{figure}[htbp]
\begin{center}
\begin{tikzpicture}[scale=0.8]
\draw[thick] (-1.5,0) node[anchor=east] {$u(X)$};
    \draw[] (-1,0) -- (3.2,0);
    \draw[] (5.3,0) -- (6.5,0);
    \draw[] (8,0) -- (9.5,0);

\draw[thick] (-0.75,0.4) node[anchor=east] {\small $-2$};
\draw[thick] (0.9,0.4) node[anchor=east] {\small $-1.5$};
\draw[thick] (3.9,0.4) node[anchor=east] {\tiny $-0.5$};
\draw[thick] (6.75,0.4) node[anchor=east] {\small $0.5$};
\draw[thick] (0.75,0) node[anchor=east] {\small $|$};

\draw[thick] (3.45,0) node[anchor=east] {\small $)$};
\draw[thick] (5.55,0) node[anchor=east] {\small $($};
\draw[thick] (5.25,0) node[anchor=east] {\small $|$};
\draw[thick] (6.75,0) node[anchor=east] {\small $)$};
\draw[thick] (8.25,0) node[anchor=east] {\small $[$};
\draw[thick] (3.75,0) node[anchor=east] {\small $|$};
\draw[thick] (-0.75,0) node[anchor=east] { $\bullet$};
\draw[thick] (2.25,0.4) node[anchor=east] {\small $-1$};
\draw[thick] (2.2,0) node[anchor=east] {\small $|$};

\draw[thick] (8.25,0.4) node[anchor=east] {\small $1$};
\draw[thick] (-0.85,0) node[anchor=east] {\small $[$};

\draw[dashed] (0.25,0) -- (1.5,-2);
\draw[dashed] (2.3,0) -- (2.6,-2);
\draw[dashed] (6.55,0) -- (8,-2);
\draw[dashed] (3.25,0) -- (4.5,-2);
\draw[dashed] (5.3,0) -- (5.6,-2);
\draw[dashed] (8,0) -- (8,-2);
\draw[thick] (5.25,0.4) node[anchor=east] {\tiny $0$};

\draw[thick] (5.25,-2.4) node[anchor=east] {\small $0$};
\draw[thick] (5.25,-2) node[anchor=east] { $\bullet$};

\draw[thick] (-1.5,-2) node[anchor=east] {$g(u(X))$};

    \draw[] (-1,-2) -- (9.5,-2);

\draw[thick] (-0.75,-2.4) node[anchor=east] {\small $-2$};
\draw[thick] (2.2,-2.4) node[anchor=east] {\small $-1$};
\draw[thick] (8.15,-2) node[anchor=east] {\small $)$};
\draw[thick] (-0.75,-2) node[anchor=east] { $\bullet$};
\draw[thick] (2.25,-2) node[anchor=east] { $\bullet$};
\draw[thick] (8.2,-2) node[anchor=east] {\small $[$};

\draw[thick] (8.25,-2.4) node[anchor=east] {\small $1$};
\draw[thick] (8.25,-2) node[anchor=east] { $\bullet$};

\end{tikzpicture}
\caption{Ilustration of case $(b_{22})$ of Corollary~\ref{Cexp}, for the particular values of $r=0.5$ and $u(a)=1$.}
\label{figureCorollaryExampleb2}
\end{center}
\end{figure}

\end{proof}

The concept of $\epsilon$-continuity generalized the idea of continuity for SS-representations of semiorders and it is used for the proofs of the conjectures. This concept was introduced for first time in \cite{liburu} as follows.

\begin{definition}\rm
Let $\prec$ be a semiorder on $(X,\tau)$. We shall say that the semiorder is \emph{$r$-continuous} (for a positive value $r\in \mathbb{R}$) if there exists a SS-representation   $(u,1)$ such that the length of each jump-discontinuity is strictly smaller than this constant $r$.


We shall say that the semiorder is \emph{$\epsilon$-continuous} if for any $\epsilon>0 $ there exists a SS-representation   $(u_{\epsilon},1)$ such that the length of each jump-discontinuity is strictly smaller than the value $\epsilon$.
\end{definition}

As it is said, this concept is weaker than the usual continuity. In \cite{liburu} it is shown that necessary conditions (NC) for the usual continuity are not needed for the existence of an $\epsilon$-continuous SS-representation. Thus,  if a semiorder is continuously SS-representable, then it is $\epsilon$-continuous, however, there exist  $\epsilon$-continuous semiorders that fail to be  continuously SS-representable. 
Furthermore, there exist semiorders that fail to be  $\epsilon_0$-continuously SS-representable, for a given $\epsilon_0>0$ (with $\epsilon_0\leq 1$). \cite{liburu}

\section{Debreu's Open Gap Lemma for Bounded Semiorders: continuous SS-representability}\label{sD}

First, in order to simplify the proofs, in this section we shall argue on \emph{irreducible} semiorders.

\begin{definition}\rm
Let $X$ be a nonempty set and $\prec$ a semiorder on $X$. We say that the semiorder is \emph{irreducible} on $X$ if there is no partition $X_1\cup X_2$ of $X$ such that $x_1\prec x_2$ for any $x_1\in X_1$ and $x_2\in X_2$.
\end{definition}

We reduce our study to those semiorders since any other one may be studied and represented through its irreducible components. In fact, given a representable   semiorder  $\prec$ on $X$ such that $X_1\cup X_2$ is a partition of $X$ satisfying that $x_1\prec x_2$ for any $x_1\in X_1$ and $x_2\in X_2$, if we know two SS-representation $(u_1,1)$ and $(u_2,1)$ of $(X_1,\prec)$ and $(X_2,\prec)$ (respectively), then it is easy to construct a representation $(u,1)$ of the semiorder as follows:
\begin{center}
$u(x) =  \left\{  \begin{array}{lcl}
u_1(x) &;& x \in X_1, \\
u_2(x)+m &;& x \in X_2, \\ 
\end{array}\right.\medskip$
 \end{center}
 where $m= \sup u_1(X_1)-\inf u_2(X_2)+2.$
 
 \begin{remark}\label{Rirre}\rm
When dealing with continuity and non irreducible semiorders it is important to take into account the necessary conditions, in particular the first condition which indicates that the semiorder must be $\tau$-continuous. In fact, given a representable semiorder  $\prec$ on $(X,\tau)$ such that $X_1\cup X_2$ is a partition of $X$ satisfying that $x_1\prec x_2$ for any $x_1\in X_1$ and $x_2\in X_2$, if we know two continuous SS-representation $(u_1,1)$ and $(u_2,1)$ of $(X_1,\prec)$ and $(X_2,\prec)$, then the aforementioned representation $u$ is continuous too in case the semiorder $\prec $ is $\tau$-continuous, i.e. in case $X_1$ and $X_2$ are open in $(X,\tau)$.
 \end{remark}

Hence, in order to simplify the present work and to avoid redundancies, from now, we shall assume that the semiorder studied is irreducible.
Before we introduce our main results, we recover the following concept and proposition introduced in \cite{liburu}.

\begin{definition}\rm
Let $(X, \tau)$ be a topological space and $u\colon X \to \mathbb{R}$ a real function on $X$. Let $I=[a,b]$ be a bounded interval of the real line. A subset $\mathcal{C}=u(X)\cap I$ is said to be a \emph{discontinuous Cantor set} if it satisfies the following properties:
\begin{enumerate}
\item[($i$)] It has measure 0,
\item[($ii$)] it has an infinite number of gaps,
\item[($iii$)] every gap of $\mathcal{C}$ is a bad gap.
\end{enumerate}

If there is a bounded interval $I$ such that $\mathcal{C}=u(X)\cap I$ is a discontinuous Cantor set, then we will say that \emph{$u(X)$ contains a discontinuous Cantor set}. 
\end{definition}

\begin{remark}\rm
Notice that, given a discontinuous Cantor set $\mathcal{C}=I\cap u(X)$, 
 then the sum of all the gaps of $\mathcal{C}$ is the length of the interval $I$.
\end{remark}

\begin{proposition}\label{Pgap}
Let $\prec$ be a bounded semiorder on $(X,\tau)$. 
 Let $(u,1)$ be a SS-representation. Then, there is no sequence of gaps $\{g_n\}_{n\in \mathbb{N}}$ such that the length of a gap $g_n$ is strictly smaller than the length of $g_{n+1}$, for any $n\in \mathbb{N}$. Hence,  there always exists a maximal gap, that is, a gap which length is the biggest.
\end{proposition}
\begin{proof}
Since the semiorder is bounded, so is its representation and, therefore, the sum of the length of the gaps (denoted by $ \displaystyle\sum\limits_{n\in\mathbb{N}} L(g_n) $) is finite. Hence, we conclude that  the sequence 
 $\{L(g_n)\}_{n\in \mathbb{N}}$ converges to 0. In consecuence,  there always exists a maximal gap.
\end{proof}

The next corollary is directly deduced from the proposition before.

\begin{corollary}\label{Cgap}
Let $\prec$ be a bounded semiorder on $(X,\tau)$. 
 Let $(u,1)$ be a SS-representa\-tion. Then, for any gap $g$ there exists another smaller gap $g'$ such that there is no gap which length is strictly between the length of  $g$ and that of $g'$. 
\end{corollary}


Now, we are ready to present our main theorems.
First, we introduce the weakest one. 

\begin{theorem}{\rm(\textbf{{The Weakest Theorem}})}\label{Tweakest}

\noindent
Let $\prec$ be a SS-representable and bounded semiorder on a topological space $(X,\tau)$ and $(u,1)$ a SS-representation. If it satifies the necessary conditions (NC) and there is no discontinuous Cantor set contained in $u(X)$, then it is $\epsilon$-continuously representable.
\end{theorem}

\begin{proof}
Let $\precsim$ be a SS-representable and bounded semiorder on a topological space $(X,\tau)$.
By Proposition~\ref{Pgap} there is a maximal gap $G_1$ that generates a discontinuity at a point $a_1$. Let $\delta_1$ be the length of $G_1$. By Proposition~4.6 in \cite{liburu}, 
 $G_1$ is of the form $ (u(a_1), r]$ or $[r, u(a_1))$.
Without loss of generality, we may suppose that $G_1=[r, u(a_1))\subseteq (0,1]$ with $u(a_1)=1$ and $r=1-\delta_1$, where $1>\delta_1>0$ is the length of the jump-discontinuity (it would be proved dually for $ (u(a_1), r]$),  otherwise we should traslate the set by a function $t_1(x)=x+1-u(a_1)$. 
Since $u(X)$ is bounded, we may divide it in $T$ unit subintervals, such that  $u(X)\subseteq I_{-M+1}\cup I_{-M+2}\cup \cdots\cup I_1\cup \cdots \cup I_N$, with $M+N=T$ and such that $[r, u(a_1))=[r,1)\subseteq [0,1]=I_1$.


\medskip


Now, focusing  on  $I_1= [0,1]$, where the biggest gap  $G_1=[r, u(a_1))$ was found (we denote the decreasing family of the lengths of the bad gaps in $I_1$ by  $(\delta_n^1)_{n\in\mathbb{N}}$), we will construct a piecewise function $f_1^1$ on $u(X)\subseteq \mathbb{R}$ that will remove this gap.\footnote{Here, the superscript 1 relates the function $f_1^1$ (the first function of a family of functions $(f_n^1)_{n\in \mathbb{N}}$) as well as the family of gaps $(\delta_n^1)_{n\in\mathbb{N}}$ to the interval $I_1$.} 
After finish the work on this subset $I_1$, we will repeat the process but now focusing on another unit interval $I_{i_2}$ related to the next biggest bad gap.  We will  denote the decreasing family of the lengths of   bad gaps in $I_{i_2}$ (with respect to the initial function $u$) by  $(\delta_n^{i_2})_{n\in\mathbb{N}}$, and so on, so that  $(\delta_n^{i_k})_{n\in\mathbb{N}}$ denotes the decreasing family of the lengths of   bad gaps in $I_{i_k}$, for each $k=1,...,T$. 

In case of a bad gap $[k-\delta_l, k+\delta_r)$ or  $  (k-\delta_l, k+\delta_r] $ 
 which is in the middle of two of those unit intervals $I_k$ and $I_{k-1}$, we will consider it in our algorithm as two consecutive bad gaps $[k-\delta_l,k)$ and $[k, k+\delta_r)$   (dually, $  (k-\delta_l, k ] $ and $  (k, k+\delta_r ] $), so that $\delta_l$ and $\delta_r$ are elements of the sequences of lengths of $(\delta_n^{k-1})_{n\in \mathbb{N}}$ and  $(\delta_n^{k})_{n\in \mathbb{N}}$, respectively. 

Coming back to $I_1$, fisrt, we define the corresponding sub-functions $\lambda^1_1, \lambda_2^1, \lambda_3^1$ and $c^1$, which are linear functions that will be applied adequately in each threshold interval in order to keep the rigid structure of the semiorder, i.e., in order to achieve another SS-representation (but now without the gap $G_1$). 

By Corollary~\ref{Cexp}, if  $ [ r-1, u(a_1)-1]$ or   $ [ r-1, u(a_1)-1)$ is a gap, then $u(X)\cap  [ r-2, u(a_1)-2]$ has at most one point $s$ (which is in fact the adjoint point $u(a_1)-2$ in case $u(a_1)-1 $ exists), and we will continue applying the cases $(a_2)$ or $(b_2)$ of Corollary~\ref{Cexp}, until arrive to a $m\in \mathbb{N}$ such that  $ u(X)\cap [r-m-\gamma_l, u(a_1)-m+\gamma_r] $ (with $\gamma_l, \gamma_r \geq 0$ and, at least, one of them positive, which also depends on the existence of those single points described in Corollary~\ref{Cexp}) has at most  one point (the point $u(a_1)-m$ in case $u(a_1)-m+1$ exists). 
The existence of this bigger gap (which does not suppose --by hypothesis, since the previus and smaller one $G_1$ we supposed to be the biggest-- a discontinuity) may allow the existence of elements on the following intervals $ [r-n, u(a_1)-n] $ for $n>m$. 

So, now, we define the following \emph{expansion}  function on the left side of the gap (i.e., for $x<r$):
 
\begin{enumerate}

\item $\lambda_1^1(x)=(x+n)\cdot \frac{1}{1-\delta_1^1}-n,$ 
$\, x\in [-n, r-n]$, $ n\in \mathbb{N} $ with $ 0\leq n<m$.

\item $\lambda_2^1(x)=(x+n)\cdot \frac{1}{1-\delta_1^1}-n, $  
 $\, x\in  [-n+\gamma_r, r-n-\gamma_l] $, $ n\in \mathbb{N} $ with $ n\geq m$. 

\item $\lambda_3^1(x)=1-n, $  
 $\, x\in  (r-n,1-n)$, $ n\in \mathbb{N} $ with $ 0\leq n<m$.


\end{enumerate}

And a \emph{contraction} function on $ [r-n-\gamma_l, u(a_1)-n+\gamma_r] $, $ n\in \mathbb{N} $ with $ n> m$, that reduces the length of this interval from $\delta_1^1+\gamma_l+\gamma_r$ to $\gamma_l+\gamma_r$: 
\begin{enumerate}
\item 
$c^1(x)=(x-(r-n-\gamma_l))\cdot \frac{\alpha_l+\alpha_2}{\delta_1^1+\gamma_l+\gamma_2}+\lambda_2^1(r-n-\gamma_l),$ 
where $\alpha_l=-\lambda_2^1(r-n-\gamma_l)$ and $\alpha_r=\lambda_2^1(1-n+\gamma_r)$.

\end{enumerate}

Notice that $(r-n,1-n)\cap u(X)$ (for $0\leq n<m$) would contain at most one point, so $\lambda_3^1$ does not imply a contraction on$(r-n,1-n)\cap u(X)$. 

\medskip

Through a dual study on the right side of the gap, we define the following \emph{expansion} function:

\begin{enumerate}
\item  $\lambda_1^1(x)=(x-n)\cdot \frac{1}{1-\delta_1^1}+n,$
 $\, x\in  [n, n+r ]$, $ n\in \mathbb{N} $ with $0<  n<m',$

\item $\lambda_2^1(x)=(x-n)\cdot \frac{1}{1-\delta_1^1}+n,$ 
 $\, x\in [1+n+\gamma_r', 1+n+r-\gamma_l' ] $, $ n\in \mathbb{N} $ with $ n\geq m'.$ 
\item  $\lambda_3^1(x)=n+1,$
 $\, x\in ( n+r , n+1) $, $ n\in \mathbb{N} $ with $0<  n<m'.$
  

\end{enumerate}

Remember  that in $ [ n+r , n+1]$ (for $0< n<m'$) there is --at most-- one point $s$ (see Corollary~\ref{Cexp}), so $\lambda_3^1$ does not imply a contraction on $u(X)\cap ( n+r , n+1) $.

And the \emph{contraction} function on $ [ r+n-\gamma_l', n+1+\gamma_r'] $, $ n\in \mathbb{N} $ with $ n\geq m'$, that reduces the length of this interval from $\delta_1^1+\gamma_l+\gamma_r$ to $\gamma_l+\gamma_r$. Here, if there is no point $s$ in $ [ r+m'-\gamma_l', m'+1+\gamma_r'] $, then the contraction function is defined as follows:
\begin{enumerate}
\item $c^1(x)=(x-(r+n-\gamma_l'))\cdot \frac{\alpha_l'+\alpha_r'}{\delta_1^1+\gamma_l'+\gamma_r'}+\lambda_2^1(r+n-\gamma_l'),$
where $\alpha_l'=-\lambda_2^1(r+n-\gamma_l')$ and $\alpha_r'=\lambda_2^1(1+n+\gamma_r')$.

\end{enumerate} 
However, if there is a point $s$ in $ [ r+m'-\gamma_l', m'+1+\gamma_r'] $ (see Corollary~\ref{Cexp}), then $c^1(x)$ is defined as before except for a possible point:

\begin{enumerate}
\item 
If $\delta_r=0$, then $\delta_l>0$ and $(s+1, m'+2]\cap u(X)=\emptyset$. 
Then,
$c^1_1(x)=(x-(r+n-\gamma_l'))\cdot \frac{\alpha_l'}{s-r-n+\gamma_l'}+\lambda_2^1(r+n-\gamma_l'),$ for any $ x\in [ r+n-\gamma_l', s+n-m'], $ and  $(s+1, n+1]$ is contracted to $n+1$,m with $n=m',m'+1$.
We shall continue arguing to the right with $(s_1+1,m'+1]$  as a gap of case (ci)
\item 
If $\delta_l=0$, then $\delta_r>0$ and $(r+m'+1, s+1)\cap u(X) $ may contain a unique point $s'$. 
Then, 
$c^1_2(x)=(x-(s+1))\cdot \frac{\alpha_r'}{n+1+\gamma_r'-s}+n$ 
 for any $ x\in [ s+n-m', n+\gamma_r' ], $ and 
$(r+n, s+n-m')$  is contracted to $n$, with $n=m',m'+1$.
We shall continue arguing to the right with $[r+2,s')$ as a gap of case (ci), in case $s'$ exists, and without restrictions otherwise.
\item If $\delta_l>0$ and $\delta_r>0$, then we apply $c^1_1(x)$ in $[ r+n-\gamma_l', s+n-m']$ and $c^1_2(x)$ in $[ s+n-m', n+\gamma_r' ].$
\end{enumerate}

We will argue similarly in case of a bad gap of the form  $G_1=( u(a_1), r]$, 
 constructing the corresponding functions.
 
It is straightforward to see that $\lambda^1_i(t)+1=\lambda^1_i(t+1)$ (for any $i=1,2,3$) as well as $c^1_k(t)+1=c^1_k(t+1)$ (for any $k=1,2$) for any $t\in u(X)$. Hence, after applying this piecewise function $f_1^1$ on $u(X)$, another SS-representation $(u_1^2, 1)$ is achieved, but now without the aforementioned gap $G_1$.


 By Corollary~\ref{Cgap}, the next biggest gap $G_2$ in $I_1$ --which length with respect to $u$ is $\delta_2^1$-- is selected and we continue with the process, achieving another representation $(u^3,1)$.

\medskip

Let's see that, given any $\epsilon_0>0$, this process arrive to a point  such that the represesentation $(u_1^n,1)$ is $\epsilon_0$-continuous in  $I_1=[0,1]$.

Let $\{\delta^1_n\}_{n\in \mathbb{N}}$ be the sequence of lengths corresponding to  the bad gaps $\{G_n\}_{n\in \mathbb{N}}$ (ordered from bigger to smaller, see Corollary~\ref{Cgap}) associated to the initial function $u=u_1$ of the SS-representation $(u, 1)$ on $I_1=[0,1]$. Obviously, $\delta_n^1$ tends to 0 when $n$ tends to infinity (see Proposition~\ref{Pgap}). Furthermore,  notice that, since there is no discontinuous Cantor set contained in $u(X)$, the sum $\sum_{k=1}^{+\infty} \delta^1_k$ is strictly smaller than 1 (i.e. there exists $r_1>0$ such that $\sum_{k=1}^{+\infty} \delta^1_k \leq 1-r_1$). 

We denote the family of bad gaps in $I_1$ corresponding to the representation $u_1^2=f_1^1\circ u_1$ by $\{f_1(G_n)\}_{n\in \mathbb{N}\setminus \{1\} }$ $= \{G_n^2\}_{n\in \mathbb{N}} .$ 
The length of the biggest gap $G_1^2$ corresponds to the expansion of the  gap $G_2$ of $u_1(X)$, thus, the length of the  gap $G_1^2$ is $l_2= \frac{\delta_2^1}{1-\delta_1^1}$.
This will be repeated again and again with each function $f_2^1, f_3^1,..., f_n^1$. 

Therefore, we are able to define the sequence $\{l_n\}_{n\in \mathbb{N}}$ of   lengths associated to the biggest bad  gap of  each representation $(u_1^n,1)$ in $I_1$.  This sequence is defined recursively as follows:

$$ l_1=\delta_1^1$$
$$l_2=\delta_2^1\cdot \frac{1}{1-\delta_1^1}=\frac{\delta_2^1}{1-l_1}$$
$$l_3= \delta_3^1\cdot\frac{1}{1-\delta_1^1}\cdot  \frac{1  }{1-\frac{\delta_2^1}{1-\delta_1^1}}  =\frac{\delta_3^1}{1-\delta^1_1-\delta^1_2}=\frac{\delta_3^1}{(1-l_1)\cdot (1-l_2)} $$
$$...$$
$$  l_n=\frac{\delta_n^1}{1- \sum_{k=1}^{n-1} \delta^1_k }=\frac{\delta_{n}^1}{\Pi_{k=1}^{n-1}(1-l_k)}.$$

Thus, since (as said before)  the sum $\sum_{k=1}^{+\infty} \delta^1_k$ is strictly smaller than 1 and $\{\delta_n^1\}_{n\in \mathbb{N}}$ tends to 0 when $n$ tends to infinity, we conclude that $ \{l_n\}_{n\in \mathbb{N}}$ tends to 0 when $n$ tends to infinity.

Therefore, for any $\epsilon_1> 0$ there is always a finite number $n_0$ such that the length of the biggest bad gap   on  $ f_{n_0}^1\circ f_{n_0-1}^1\circ \cdots \circ f_1^1 ( u(X)\cap [0,1])$ is smaller than  $\epsilon_1$ (we shall denote $f^1=f_{n_0}^1\circ f_{n_0-1}^1\circ \cdots \circ f_1^1$). In the limit, we reduce the measure of the union of bad gaps to 0, stretching $\frac{1}{1-\sum\limits_{n\in \mathbb{N}} \delta^1_n}$ times the subset $u(X)\cap [0,1]$.

After achieve the desired result on 
 $[0,1]$, we choose now the next biggest gap (if it exists) which lies in $I_{i_2}$, for some $i_2\in \{-M+1,...,N\}$. Again (because of traslations)  we may assume that it is of the form  $[u_2(a_2)-\delta_1^2, u_2(a_2))\subseteq [0, 1] $ (or the dual $(u_2(a_2), u_2(a_2)+\delta_1^2]\subseteq [0, 1])$), where $u_2$ now denotes the function $ f_{n_0}^1\circ f_{n_0-1}^1\circ \cdots \circ f_1^1 \circ u$. Now, we repeat the same process on $ [u_2(a_2)-1, u_2(a_2))$ as before, achieving a family of functions that are applied on $u_2(X)$. After that, the aforementioned bad gaps of 
 the first interval $I_1= [0,1]$ (where the biggest gap  $G_1=[r, u(a_1))$ was found)
 may increase 
 $\frac{1}{1-\sum\limits_{n\in \mathbb{N}} \delta^1_n}\cdot \frac{1}{1-\sum\limits_{n\in \mathbb{N}} \delta^{i_2}_n}$ times, where now $\sum\limits_{n\in \mathbb{N}} \delta^{i_2}_n$ is the sum of the length of the bad gaps in  $I_{i_2}$ with respect to the first function $u$. 
  Hence, if the biggest discontinuity desired is $\epsilon_0$, we should choose the $\epsilon_1$ before smaller than  $\epsilon_0\cdot (1-\sum\limits_{n\in \mathbb{N}} \delta^1_n) \cdot(1-\sum\limits_{n\in \mathbb{N}} \delta^{i_2}_n)$. As a matter of fact, since each $k^{th}$ step may increase the length of the gaps of the intervals corresponding to the steps before $\frac{1}{1-\sum\limits_{n\in \mathbb{N}} \delta^{i_k}_n}$ times,  we should choose the $\epsilon_1$ before such that  $$\epsilon_1 <\epsilon_0\cdot \prod\limits_{k=1}^T (1-\sum\limits_{n\in \mathbb{N}} \delta^{i_k}_n).$$

Since the semiorder is bounded, the image of $X$ is contained in a bounded interval so, the process ends up after a finite number $T$ of steps, achieving a SS-representation $(u_{T+1}, 1)$ where the length of the biggest bad gap is less than the desired value $\epsilon_0>0$.
\end{proof}

Furthermore, we may apply infinite steps in the proof before, i.e., it is possible to construct a sequence of SS-representations that converges to a limit. The following theorem shows that this limit exists as well as it is in fact a continuous SS-representation.

\begin{theorem}{\rm(\textbf{{The Weak Theorem}})}

\noindent
Let $\prec$ be a SS-representable and bounded semiorder on a topological space $(X,\tau)$ and $(u,1)$ a SS-representation. If it satifies the necessary conditions (NC) and there is no discontinuous Cantor set contained in $u(X)$, then it is continuously representable.

\end{theorem}

\begin{proof}
If it satifies the necessary conditions (NC), then it is $\epsilon$-continuously representable. Hence, for any $n\in \mathbb{N}$ there exists a SS-representation $(u_n,1)$ such that the length of the biggest gap is less than $\frac{1}{n}$. As a matter of a fact, given $\frac{1}{n_0}$, in the previus proof we used a method to construct a SS-representation $(u_{n_0}, 1)$ which is $\frac{1}{n_0}$-continuous. Hence, for  $\frac{1}{n_0+1},$ we proceed analogously but now starting the process from $u_{n_0}$.

In consecuence, we are able to construct a sequence of functions $\{u_n\}_{n\in \mathbb{N}}$ which are (respectively) $\frac{1}{n}$-continuous, for any $n\in \mathbb{N}$. Furthermore,  this is  a pointwise Cauchy sequence (with respect to the supremum norm). 
To see this, first notice that $u_{n+1}(x)=f^n u_n(x)$ so, $ ||u_{n+1}-u_n||_{\infty}=\sup\{u_{n+1}(x)-u_n(x)\}_{x\in X} $ is just
$$ ||f^n\circ u_{n}-u_n||_{\infty}=\sup\{f^n(u_{n}(x))-u_n(x)\}_{x\in X} =\sup\{f^n(r)-r\}_{r\in u_n(X)}.$$
Now, remember that $f^n$ is a piecewice function defined by the linear functions $\lambda_1^n, \, \lambda_2^n, \, \lambda_3^n$ and $c^n$ (see the proof of Theorem~\ref{Tweakest}). The slopes of those linear functions depends on the lengths of the gaps, so that they tend to 1  when $n$ tend to infinite. Thus, it is straigthforward to check that those linear functions tend to identity when $n$ increase. Therefore, it holds true that $\lim \limits_{n\to \infty}\sup\{f^n(r)-r\}_{r\in u_n(X)}=0$ and, we conclude that  $\{u_n\}_{n\in \mathbb{N}}$ is a pointwise Cauchy sequence of SS-representations.

Therefore, there exists the limit function $u=\lim \limits_{n\to \infty} u_n$ which  is in fact   a SS-re\-presen\-ta\-tion. To see that, first notice that for any function $f^n$ applied on $u_n(X)$, it holds that \mbox{$f^n(u_{n}(x)+1)=f^n(u(x))+1$}, for any $x\in X$ and in case $u_n(x)+1 \in {u_n(X)}$. Therefore --and since $f^n$ is strictly increasing on $u_n(X)$--, the condition of being a SS-representation is also satisfied by $f^n\circ u_n= u_{n+1}$.

So, we have that $x\prec y $ if and only if $u_n(x)+1<u_n(y)$, for any $n\in \mathbb{N}$, and then $u(x)+1\leq u(y)$. Let's see that the strict inequality is also keep for the limit function $u$.

First, remember  that, since the initial SS-representation also represents the main trace $\precsim^0$, it is also represented by any $u_n$. Now we distinguish two cases:

\begin{enumerate}
\item[$(a)$] If there is no element $z\in X$ such that $x\prec^0 z \prec y$ or $x\prec z\prec^0 y$, that is, if $(x,y)$ is a gap, since $x\prec y$ we may make a partition of $X$ by $X=X_1\cup X_2$ with $X_1=L_{\precsim^0}(x)$ and $X_2= U_{\precsim^0}(y)$ such that $x_1\prec x_2$ for any $x_1\in X_1$, $x_2\in X_2$. Thus, the semiorder is not irreducible and it is the union of two semiorders, so it is out of our study now since these cases are trivial when the irreducible components of the union are known (see Remark~\ref{Rirre}).

\item[$(b)$] Otherwise, there is --at least-- one element $z\in X$ such that $x\prec^0 z \prec y$ or $x\prec z\prec^0 y$. Suppose that $x\prec^0 z \prec y$, then $u_n(x)+1<u_n(z)+1<u_n(y)$ for any $u_n$ of the sequence, hence,  if $\lim\limits_{n\to \infty} u_n(x)+ 1= \lim\limits_{n\to \infty} u_n(y)$ it means that $\lim\limits_{n\to \infty}u_n(x)=\lim\limits_{n\to \infty}u_n(z)$. Thus, some points have been contracted to a single point. We argue similarly for $x\prec z\prec^0 y$. 
Therefore, these points must be affected by the contraction functions $c^n$. However, since the semiorder is bounded, the number of contraction functions that may be applied in the same element is finite. Hence, the equality  $\lim\limits_{n\to \infty} u_n(x)+ 1= \lim\limits_{n\to \infty} u_n(y)$ is impossible.
 \end{enumerate}

 Thus, we conclude that $u=\lim\limits_{n\to \infty} u_n$ is a SS-representation too.

 Finally, since the length  of the biggest gap of $u$ is less than $\frac{1}{n}$ for any $n\in \mathbb{N}$, we conclude that $u$ is continuous.
\end{proof}

Now, we present the \emph{Strong Theorem}, where Debreu's Open Gap Lemma is needed for the proof. Here, the absence of Cantor subsets is not required.


\begin{theorem}{\rm(\textbf{The Strong Theorem})}\label{Tstrong}

\noindent
Let $\prec$ be a SS-representable and bounded semiorder on a topological space $(X,\tau)$. If it satifies the necessary conditions (NC), then it is continuously representable.
\end{theorem}

\begin{proof}
Let $(u,1)$ be a SS-representation of the semiorder. We assume, without loss of generality, that $u$ also represents the total preorder $\precsim_0$.
We shall use  function $g$ of Debreu's Open Gap Lemma, which removes the bad gaps of $S$ achieving another set $g(S)$ of the same length but now without bad gaps.

Since the semiorder is bounded, we may assume without loss of generality that $-M+w$ is the infimun of $u(X)$ and $N$ is the supremum, where $N,M\in \mathbb{N}$ and $w\in [0,1)$.
Let $\overline{u(X)}$ be the set defined by $[-M+w,N]\setminus \bigcup_{n\in \mathbb{N}}G_n$, where $\{G_n\}_{n\in \mathbb{N}}$ is the family of bad gaps of $u(X)$.

By Proposition~\ref{Pgap}, there is a maximal gap $G_1$ on $u(X)$ of the form $ (u(a_1), r]$ or $[r, u(a_1))$. Without loss of generality (and with the only purpose of simplify notation), we may suppose that $G_1=[r, u(a_1))$ with $u(a_1)=1$ and $r=u(a)-\delta_1$,  otherwise we would translate the set $u(X)$. 

\medskip

First, we focus on $I_0=[u(a_1)-1, u(a_1)]=[0,1]$. Here, we apply function $g$ that removes the bad gaps of $\overline{u(X)}$ in $I_0$, returning a subset $S=g(\overline{u(X)}\cap [0,1])$ free of bad gaps and such that $\inf S=0$ and $\sup S=1$.
Notice that given that biggest bad gap  $G_1=[r, u(a_1))=[r,1)$, by Corollary~\ref{Cexp}, there is no bad gap containing 0. That is, there is no bad gap $(s,t]$ or $[s,t)$ with $s<0<t$. 
 Thus, $\inf\{ \overline{u(X)}\cap [0,1]\}=0 $ and $g$ is defined from 0 to 1 (with some possible gaps in the middle in addition to $G_1$).
The changes made on $I_0$ must be taking into account in $I_1=[1,2]$ and $I_{-1}=[-1,0]$ (if they exist) in order to keep the semiorder relation, so now we apply on 
 $I_1=[1,2]$ and $I_{-1}=[-1,0]$ the functions\footnote{This notation is deboted to help on the understanding of the meaning of the corresponding function. Hence, $g_0^1$ 
  makes reference to the function $g_0^0$ (already defined on $I_0$) modified to be applied on $I_1$.} $g_0^1$ and $g_0^{-1}$.

 Before define $g_0^1$ and $g_0^{-1}$, first notice that, by Corollary~\ref{Cexp}, if $[a,b)$ (dually $(a,b]$) is a bad gap in $I_0$, then   $[a-1,b-1)$ (respectively $(a+1,b+1]$) is a lacuna, so the same function $g_0^0$ can be applied in $I_1$ and $I_{-1}$, since there is no point $u(x)$ in the interval that could be removed, i.e. $g_0^1$ and $g_0^{-1}$ are strictly increasing on $u(X)\cap I_{1}$ and $u(X)\cap I_{-1}$, respectively. 
Thus, we are able to define $g_0^1(x)=g_0(x-1)+1$ (for any $x\in I_1=[1,2]$) and  $g_0^{-1}(x)=g_0(x+1)-1$ (for any $x\in I_{-1}=[-1,0]$).

  Reasoning analogously on $I_2$ and $I_{-2}$ with $g_0^{1}$ and $g_0^{-1}$, respectively, we define the functions $g_0^{1,2}$ and $g_0^{-1,-2}$ on $I_2$ and $I_{-2}$, and so on,  until arrive from the left to a $m\in \mathbb{N}$ such that  $ u(X)\cap [r-m-\gamma_l, u(a_1)-m+\gamma_r] $ (with $\gamma_l, \gamma_r \geq 0$ and, at least, one of them positive) has at most one point. 
%
 Dually, from the right, until arrive   to   $n\in \mathbb{N}$ such that  $ u(X)\cap [r+n-\gamma_l', u(a_1)+n+\gamma_r'] $ (with $\gamma_l', \gamma_r' \geq 0$ and, at least, one of them positive) has at most one point. %
%
  Here we apply the corresponding functions $g_0^{1,2,...,n}$ and $g_0^{-1,-2,...,-m}$. Before continue applying these functions in the successive intervals $I_{-m'}, I_n'$ (for $m'>m$ and $n'>n$), first we have to make some modifications on the set.

  At this point, first we focus on $ u(X)\cap [r-m'-\gamma_l, u(a_1)-m'+\gamma_r] $ (for any $m'>m$).
  Assume that there is a point $s\in  u(X)\cap [r-m-\gamma_l, u(a_1)-m+\gamma_r] $.    We distinguish three cases.
  
\begin{enumerate}
\item[$(a)$] If $\gamma_r=0$, then  $\gamma_l>0$ and    $[s-1, u(a_1)-n-1]$ may contain at most one point $s'$. In this case,  we identify $[r-m-1-\gamma_l, s-1]$ with $[r-m-1-\gamma_l', r-m-1]$. 
We shall continue arguing to the left (for $m'>m+1$) with $[s',u(a_1)-m'-1)$ as a gap of case (ci), in case $s'\in u(X)$ exists, and  just identifying $[r-m'-\gamma_l', u(a_1)-m']$ with $[r-m'-\gamma_l', r-m']$ otherwise.

\item[$(b)$] If $\gamma_l=0$, then  $\gamma_r>0$ and $[s-1, u(a_1)-m-1]$ may contain at most one point $s'$. 
In this case  we identify $[s', u(a_1)-m-1+\gamma_r']$ with $[u(a_1)-m-1 , u(a_1)-m-1+\gamma_r']$.
 We shall continue arguing to the left with $[r-m-1,s')$ as a gap of case (ci), in case $s'$ exists, and   just identifying 
 $[r+m', u(a_1)+m'+\gamma_r']$ with $[u(a_1)+m' , u(a_1)+m'+\gamma_r']$ otherwise.

\item[$(c)$] Otherwise, $\gamma_r>0$ and $\gamma_l>0$ 
and we apply both identifications described in $(a)$ and $(b)$.
 That is, we identify $[r-m'-1-\gamma_l, s-1]$ with $[r-m'-1-\gamma_l', r-m'-1]$ and $[s', u(a_1)-m'-1+\gamma_r']$ with $[u(a_1)-m'-1 , u(a_1)-m'-1+\gamma_r']$.
\end{enumerate}

If there is no    point $s$ in  $ u(X)\cap [r-m-\gamma_l, u(a_1)-m+\gamma_r] $, then 
we identify $[r-m' -\gamma_l, u(a_1)-m' )$ with $[r-m'-\gamma_l, r-m']$ in case $\gamma_l>0$,  otherwise 
we identify $[r-m'  , u(a_1)-m'+\gamma_r )$
 with $[u(a_1)-m',u(a_1)-m'+\gamma_r'] $, for any $m'>m$.

  Now, we focus on the right side  $ u(X)\cap [r+n'-\gamma_l', u(a_1)+n'+\gamma_r'] $ (for any $n'>n$). Assume that there exists that  point $s$ in  $ u(X)\cap [r+n-\gamma_l', u(a_1)+n+\gamma_r'] $.   We distinguish three cases.

\begin{enumerate}
\item[$(a)$] If $\gamma_r'=0$, then  $\gamma_l'>0$ and    $[s+1, u(a_1)+n+1]$ must be empty. In this case,  we identify $[r+n+1-\gamma_l', s+1]$ with $[r+n+1-\gamma_l', r+n+1]$. 
We shall continue arguing to the right (for $n'>n+1$) with $[r+n+1,s+2)$ as a gap of case (ci), in case $s+2\in u(X)$ exists, and  just identifying $[r+n'-\gamma_l', u(a_1)+n']$ with $[r+n'-\gamma_l', r+n']$ otherwise.

\item[$(b)$] If $\gamma_l'=0$, then  $\gamma_r'>0$ and $[r+n, s+1]$ may contain one point $s'$
In this case  we identify $[s+1, u(a_1)+n+1+\gamma_r']$ with $[u(a_1)+n+1 , u(a_1)+n'+\gamma_r']$.
 We shall continue arguing to the right with $[r+n+1,s')$ as a gap of case (ci), in case $s'$ exists, and   just identifying 
 $[r+n', u(a_1)+n'+\gamma_r']$ with $[u(a_1)+n' , u(a_1)+n'+\gamma_r']$ otherwise.

\item[$(c)$] Otherwise, $\gamma_r'>0$ and $\gamma_l'>0$ 
and we apply both identifications described in $(a)$ and $(b)$.
 That is, we identify $[r+n'-\gamma_l', s+n'-n]$ with $[r+n'-\gamma_l', r+n']$ and $[s+n'-n, u(a_1)+n'+\gamma_r']$ with $[u(a_1)+n' , u(a_1)+n'+\gamma_r']$, with $n>n'$.
\end{enumerate}

If there is no    point $s$ in  $ u(X)\cap [r+n-\gamma_l', u(a_1)+n+\gamma_r'] $, then 
we identify $[r+n' -\gamma_l', u(a_1)+n' )$ with $[r+n'-\gamma_l', r+n']$ in case $\gamma_l'>0$,  otherwise 
we identify $[r+n'  , u(a_1)+n'+\gamma_r' )$
 with $[u(a_1)+n',u(a_1)+n+\gamma_r'] $, for any $n'>n$.

  Now, we are able to  apply successfully the corresponding functions $g_0^{1,2,...,n}$ and $g_0^{-1,-2,...,-m}$ on  the successive intervals $I_{-m'}, I_n'$ (for $m'>m$ and $n'>n$),
  until arrive to the last intervals   $[-M, -M+1]$ and $[N-1,N]$. Hence, we have applied on $u(X)$ the piecewise function $g_0$ defined as $g_0(x)=g_0^{0}(x)$ if $x\in I_0$, $g_0(x)=g_0^{1,...,k}(x)$ if $x\in I_k$ (for any $k=1,...,n$), and $g_0(x)=g_0^{-1,...,-k}(x)$ if $x\in I_k$ (for any $k= -1,...,-m$).  We denote now by $u_0$ the function $g_0\circ u$.

 \medskip

 By $g_0^0$ all the bad gaps on $[0,1]$ have been removed and then, through functions $g_0^1, g_0^{1,2}$, ..., $g_0^{1,2,...,n}$ and $g_0^{-1}, g_0^{-1,-2},..., g_0^{-1,-2,...,-m}$ the changes made in $[0,1]$ have been reproduced in $I_1, I_2, ..., I_n$ and $I_{-1}, ...,I_{-m}$ in order to keep the semiorder relation. 
 \medskip

 Nevertheless, we have not removed possible bad gaps in $g_0(\overline{u(X)})\setminus [0,1]$, thus, the preocess continue 
  but now focusing on the biggest bad gap $G_1^0$ of $g_0(\overline{u(X)})$\footnote{Here the superscript 0 of $G_1^0$ refers to  function $u_0$.}. 

Again, without loss of generality, we may suppose that $G_1^0=[r, u_0(a_1))$  with $u_0(a_1)=1$ and $r=u_0(a_1)-\delta_1^0$ (it would be proved dually for $ (u_0(a_1), r]$),  otherwise we should translate the set by a function $t_1(x)=x+1-u_0(a_1)$. 
Hence, we would define again function $g_1^0$ on $G_1^0=[u_0(a_1)-1, u_0(a_1))$ and then   the corresponding functions  $g_1^{ 1}, g_1^{1,2}$,..., $g_1^{1,2,..., N_1}$ and $g_1^{-1}, g_1^{-1,-2}$,..., $g_1^{-1,-2,...,-M_1}$ (as we did before), constructing the piecewise function $g_1$.
 
Since the amount of intervals is finite, the process ends up after applying a last piecewise function $g_T$ (where $T=M+N$) defined by means of a family of  functions  $g_T^{ 1}, g_T^{1,2}$,..., $g_T^{1,2,..., N_T}$ and $g_T^{-1}, g_T^{-1,-2}$,..., $g_T^{-1,-2,...,-M_T}$, achieving a continuous SS-representation.
\end{proof}

Under the assumption of the \emph{Axiom of choice}, Theorem~\ref{Tstrong} may be generalized to unbounded semiorders as follows.

\begin{corollary}
Let $\prec$ be a SS-representable  semiorder on a topological space $(X,\tau)$. If it satifies the necessary conditions (NC), then it is continuously representable.
\end{corollary}

Then, we could conclude the following result, which is directly deduced from Theorem~\ref{Tstrong} and that 
 we present as a \emph{Debreu's Open Gap Lemma for Bounded Semiorders}.

\begin{corollary}{\rm(\emph{\textbf{Debreu's Open Gap Lemma for Bounded Semiorders}})} 

\noindent Let $S$ be a bounded subset of $\mathbb{R}$. Then, there exists a strictly increasing function $g\colon S\to \mathbb{R}$ such that all the gaps of $g(S)$ are open or closed, and satisfying that  $x+1<y\iff g(x)+1<g(y)$ if and only if the following conditions hold:
\begin{enumerate}
\item[$(i)$] There are no open-closed or closed-open gaps which length is bigger than or equal to 1. 
\item[$(ii)$] For any gap $[a,b)$: 
      \begin{enumerate}
      \item[(a)] $[a+n, b+n]\cap S$ 
      may content one point $s_n$ such that $s_{n+1}\leq s_n+1$, for any $n\in\mathbb{N}$ with $n<m_r$, for some $m_r\in\mathbb{N}$ such that    there exist $\gamma_l, \gamma_r\geq 0$ (with at least one of them different from 0) satisfying     that  $S\cap [a+m_r-\gamma_l, b+m_r+\gamma_r]$ 
     contains at most one point $s$. 
      If there exists that point $s$ and $ \gamma_r>0$, then  $ (s+n,b+n]$ may be nonempty, for any $n>m_r$ and, if $ \gamma_l>0$,  then $ [a+n, s+n]$ may contain more than one point.
      \item[(b)]  $[a-n, b-n]$ or $[a-n,b-n)$ are gaps, 
       for any $n\in\mathbb{N}$ with $n<m_l$, for some $m_l\in\mathbb{N}$ such that 
there exist  $\gamma_l, \gamma_r\geq 0$ (with at least one of them different from 0) satisfying that       
       $S\cap [a-m_l-\gamma_l, b-m_l+\gamma_r]$ 
       contains at most one point $s_{m_l}$, and such that $s_{m_l}+1\geq s_{m_l-1}$ (in case $s_{m_l-1}\in S\cap [a-m_l+1-\gamma_l, b-m_l+1+\gamma_r]$).
      \end{enumerate}

\item[$(ii)$] For any gap $(a,b]$: 
      \begin{enumerate}
      \item[(a)] $S\cap [a-n,b-n]$ contains at most one point $s_n$ such that $s_{n}+1\leq s_{n+1}$, for any $n\in\mathbb{N}$ with $n<m_l$, for some $m_l\in\mathbb{N}$ such that there exist $\gamma_l, \gamma_r\geq 0$ (with at least one of them different from 0) satisfying that      
      $S\cap [a-m_l-\gamma_l, b-m_l+\gamma_r]$ may contain at most one point $s$. 
      If there exists that point $s$ and $ \gamma_r>0$, then  $ [s-n,b-n]$ may contain more than  one point, for any $n>m_l$ and, if $ \gamma_l>0$, then  $ [a-n, s-n)$ may  be nonempty.
      \item[(b)]  $[a+n, b+n]$ or $[a+n,b+n)$ are gaps, for any $n\in\mathbb{N}$ with $n<m_r$, for some $m_r\in\mathbb{N}$ such that there exist $\gamma_l, \gamma_r\geq 0$ (with at least one of them different from 0) satisfying that     $S\cap [a+m_r-\gamma_l, b+m_r+\gamma_r]$ contains at 
contains at most one point $s_{m_r}$, and such that $s_{m_r-1}+1\geq s_{m_r}$ (in case $s_{m_r-1}\in S\cap [a+m_r-1-\gamma_l, b+m_r-1+\gamma_r]$).

      \end{enumerate}
\end{enumerate}
\end{corollary}

\section{A constructive weak version  of Debreu's Open Gap Lemma}\label{sDeb}

Finally, in the present work  we also achieve as a subproduct of the section above a new proof of a \emph{weak} version of \emph{Debreu's Open Gap Lemma.}
 We referred to it as  \emph{weak} because the absence of discontinuous Cantor set is assumed.
 
Althougth the notion of $\epsilon$-continuity was defined for semiorders, it may be generalized to other kind of representations of orderings under some adequate hypothesis. For example, in the case of total preorders or interval orders, assuming -without lost of generality in the case of total preorders- that the infimum and supremum of the representation are 0 and 1 -or any other values-, respectively. Hence, the length of the biggest jump-discontinuity may be compared with the diameter (or any other invariant) of the image of the representation (i.e.with 1). 

Taking into account the notion of  $\epsilon$-continuity for total preoders, this weak version has a big remarkable benefit: the method is constructive and finite.

\begin{corollary}
Let $S$ be a subset of the extended real line $\overline{\mathbb{R}}$. Assume that there is no discontinuous Cantor set contained in $S$. Then, there exists a strictly increasing 
 function $g\colon S\to \mathbb{R}$ such that all the gaps of $g(S)$ are open or closed.

Furthermore, this function $g$ may be built as a composition of a family (may be infinite) of linear 2-piecewise functions. As a matter of a fact, for any $\epsilon>0$, this function $g$ can be constructed (in a finite number of steps) by a composition of a finite family of linear 2-piecewise\footnote{We understand by 2-piecewise function a  function which is defined by two pieces.} functions, satisfying that the length of the bad gaps of $g(S)$ is less than $\epsilon.$ 
\end{corollary}

\begin{proof}
First, since $\overline{\mathbb{R}}$ is homeomorphic to $[0,1]$ (see \cite{BR}),  without lost of generality we may assume that $S\subseteq [0,1]$, with $\sup S=1$ and $\inf S=0$. 

By Proposition~\ref{Pgap}, there is a maximal open-closed or closed-open gap (we shall refer to them as \emph{bad gaps}): $G_1$. 
Let $\{\delta_n\}_{n\in \mathbb{N}}$ be the sequence of lengths corresponding to  the bad gaps $\{G_n\}_{n\in \mathbb{N}}$ (ordered  from biggest to smaller, see Corollary~\ref{Cgap}).

Now, focusing  on  $G_1$, we will construct a   function $f_1$ on $S$ that will remove this gap.  We denote $G_1$ by $G_1=[a_1,b_1) $ or $G_1=(a_1,b_1)$. 
We define the following strictly increasing function: 

\begin{center}
$f_1(x) =  \left\{  \begin{array}{lcl}
x\cdot \frac{1}{1-\delta_1}&;& x\leq a_1, \\
(x-\delta_1)\cdot \frac{1}{1-\delta_1}&;& x\geq b_1, \\
\end{array}\right.\medskip$
 \end{center}

Notice that $f_1$ keeps the length of $S$, i.e. $\sup f_1(S)- \inf f_1(S)=1$. We repeat the process with the next gap, $f_1(G_2), $ (we may denote   $f_1(G_2)$  by $[a_2,b_2)$ or $(a_2,b_2]$. 

Therefore, we may define the sequence $\{l_n\}_{n\in \mathbb{N}}$ of   lengths associated to the biggest bad  gap  after apply each functions $f_1,...,f_n$ on $S$, respectively, i.e. $l_{n+1}$ denotes the length of the biggest gap of $f_n\circ \cdots \circ f_1 (S)$ which corresponds to $f_n\circ \cdots \circ f_1 (G_{n+1})$.  This sequence is defined recursively as follows:


$$ l_1=\delta_1$$
$$l_2=\delta_2\cdot \frac{1}{1-\delta_1}=\frac{\delta_2}{1-l_1}$$
$$l_3= \delta_3\cdot\frac{1}{1-\delta_1}\cdot  \frac{1  }{1-\frac{\delta_2}{1-\delta_1}}  =\frac{\delta_3}{1-\delta_1-\delta_2}=\frac{\delta_3}{(1-l_1)\cdot (1-l_2)} $$
$$...$$
So, it may be proved by induction that the length of the biggest gap after $n-1$ steps is
$$  l_n=\frac{\delta_n}{1- \sum_{k=1}^{n-1} \delta_k }=\frac{\delta_{n}}{\Pi_{k=1}^{n-1}(1-l_k)}.$$

Thus, since (as said before)  the sum $\sum_{k=1}^{+\infty} \delta_k$ is strictly smaller than 1 and $\{\delta_n\}_{n\in \mathbb{N}}$ tends to 0 when $n$ tends to infinity, we conclude that $ \{l_n\}_{n\in \mathbb{N}}$ tends to 0 when $n$ tends to infinity.

Therefore, we have constructed a sequence of strictly increasing functions $\{g_n\}_{n\in \mathbb{N}}$, where $g_n= f_n\circ \cdots \circ f_1 $, such that $g_n$ is $\epsilon_n$-continuous, for a family of positive values $\{\epsilon_n\}_{n\in \mathbb{N}}$ such that $\epsilon_n$ tend to 0 when $n$ increases. As a matter of fact, since the sequence $\{\delta_n\}_{n\in \mathbb{N}}$ is decreasing and converges to 0, notice that the sequence of functions $\{g_n\}_{n\in \mathbb{N}}$ is a pointwise Cauchy sequence. To see that, notice that --since $\{\delta_n\}_{n\in \mathbb{N}}$ is decreasing and converges to 0-- the sequence of functions $\{f_n\}_{n\in \mathbb{N}}$ converge to the identity function, and since $g_{n+1}=f_{n+1}\circ g_n$, the Cauchy property is deduced.  Hence, there exists the limit function $g$ of $\{g_n\}_{n\in \mathbb{N}}$. 

Let's see that $g$ is continuous as well as strictly increasing.

On one hand, for any $\epsilon>0$, there is $n_\epsilon\in \mathbb{N}$ such that the length of the bad gaps of $g_n(S)$ are smaller than $\epsilon$, for any $n\geq n_\epsilon$. Hence, the bad gaps of $g(S)$ have no positive length, i.e. there are no bad gaps in $g(S)$.

On the other hand, $g $ is still strictly increasing. To see that, first notice that if two point $x,y\in S$ are in the same side of a gap, i.e. $x,y\leq a_1$ or $x,y\geq b_1$ for a gap $G_1=[a_1,b_1)$ or  $G_1=(a_1,b_1]$, then after apply $f_1$ the distance between $x,y$ increases such that $d(f_1(x),f_1(y))=\frac{d(x,y)}{1-\delta_1}$, where $1>\delta_1>0$ is the length of the gap.  
The distance between $x$ and $y$ is reduced just in case $x\leq a_1$ and $y\geq b_1$ (or viceversa). In that case, it holds that $d(f_1(x),f_1(y))=\frac{1}{1-\delta_1}\cdot (d(x,y)-\delta_1)$.
It can be proved that after repeat this contraction process $n$ times (i.e. after apply function $g_n= f_n\circ \cdots \circ f_1 $) the distance achieved is
$$  d(g_n(x),g_n(y))=\frac{1}{1-\sum_{k=1}^{n} \delta_k}\cdot (d(x,y)-\sum_{k=1}^{n} \delta_k) .$$
In fact, we may argue by induction and assume that after $n$ steps, the distance is as described before.
Hence, in the next step, for $n+1$, the distance would be

$$  d(g_{n+1}(x),g_{n+1}(y))=\frac{1}{1-l_{n+1}}  \cdot (d(g_n(x),g_n(y))-l_{n+1}),$$
where $l_{n+1}$ is the length of the biggest gap between $g_n(x)$ and $g_n(y))$ and that comes from the transformation of the initial gap of length $\delta_n$  through the $n$ previous steps. Hence, $l_{n+1}$ is as described before, $ l_{n+1}=\frac{\delta_{n+1}}{1- \sum_{k=1}^{n} \delta_k }.$
Thefore, replacing $l_{n+1}$ by  $\frac{\delta_{n+1}}{1- \sum_{k=1}^{n} \delta_k }$ in the equation:

$$  d(g_{n+1}(x),g_{n+1}(y))=\frac{1}{1-\frac{\delta_{n+1}}{1- \sum_{k=1}^{n} \delta_k }}  \cdot (d(g_n(x),g_n(y))-\frac{\delta_{n+1}}{1- \sum_{k=1}^{n} \delta_k }).$$

Finally, after replace the value $  d(g_n(x),g_n(y))=\frac{1}{1-\sum_{k=1}^{n} \delta_k}\cdot (d(x,y)-\sum_{k=1}^{n} \delta_k)$ and simplify, we achieve the desired result:

$$  d(g_{n+1}(x),g_{n+1}(y))=\frac{1}{1-\sum_{k=1}^{{n+1}} \delta_k}\cdot (d(x,y)-\sum_{k=1}^{{n+1}} \delta_k) .$$

Thus, given any two points $x<y$ in $S$, in the limit --i.e. after apply function $g$-- the distance between  $g(x)$ and $g(y)$  is at least, as bigger as

$$  \frac{1}{1-\sum_{k=1}^{+\infty} \delta_k}\cdot (d(x,y)-\sum_{k=1}^{+\infty} \delta_k) .$$

By hypothesis, there is no discontinuous Cantor set contained in $S$, hence, $\sum_{k=1}^{+\infty} \delta_k$ is stricly smaller than $d(x,y)$ (and, in particular,   stricly smaller than 1).
Thus,  $g(x)<g(y)$ for any $x<y$. This concludes the proof.  
\end{proof}

\begin{remark}\rm
Given $\epsilon_0>0$,
 the proof before can be limited to a finite number of steps just in order to achieve $\epsilon_0$-continuity through a constructive and finite process. This may be interesting for programming purposes.
\end{remark}






\section{Concluding remarks}\label{s5}
Since a semiorder rarely would be continuously representable, 
  the concept of $\epsilon$-continuity seems a successful  tool when dealing with semiorders that fail to be continuously representable, but also for those that may be continuously represented.

 Furthermore, 
 %
  through this idea we have presented a characterization for the continuous  SS-representability of bounded semiorders. For any $\epsilon_0>0$, the construction of the desired $\epsilon_0$-continuous  representation may be done in a finite number of steps, and this may be interesting for programming purposes. 
 
 Hence, analogous result of the famous  \emph{Debreu's Open Gap Lemma}  but for bounded semiorders has been proved. That is, now, given any bounded subset  $S\subseteq \mathbb{R}$,
we chacarterized when  there exists a strictly increasing function
$g\colon S\to \mathbb{R}$ such that all the gaps of
$g(S)$ are open as well as $g$  satisfies the geometrical condition
 $x \prec y \Leftrightarrow g(x) + 1 < g(y)$, for every $x,y \in X$.

 Some other proof for these results may exist, in particular, by induction on the length of $u(X)$ (i.e., on the length of the longest chain $x_1\prec x_2\prec \cdots \prec x_l$) or aggregating continuous representations $u_1\colon (X_1,\tau_{|X_1})\to (\mathbb{R}, \tau_u)$ and $u_2\colon (X_2,\tau_{|X_2})\to (\mathbb{R}, \tau_u)$ of a semiorder on $(X,\tau)$, with $X=X_1\cup X_2$.



\end{document}